\documentclass[letterpaper,10pt,conference]{ieeeconf}
\usepackage{color}
\usepackage[dvipsnames]{xcolor}
\usepackage{amsmath}
\usepackage{amssymb}
\usepackage{graphicx}
\usepackage{bm}
\usepackage{url}
\usepackage{algorithm}
\usepackage{algpseudocode}
\usepackage{epsfig} 
\usepackage{color}
\usepackage[normalem]{ulem}
\usepackage{cancel}
\usepackage{fixltx2e}
\usepackage{array}
\usepackage{booktabs}

\usepackage{ams fonts} % need for calligraphy fonts (e.g. Range, Null space)
\usepackage{xcolor}
\usepackage{epstopdf}
\usepackage{url}
\usepackage{pifont}
\allowdisplaybreaks

\newcommand{\norm}[1]{\left\lVert#1\right\rVert}

\newtheorem{lem}{\bf{Lemma}}
\newtheorem{Def}{\bf{Definition}}
\newtheorem{prop}{\bf{Proposition}}
\newtheorem{thm}{\bf{Theorem}}

\begin{document}


%\title{\Large \vspace{1.8cm}
%	Secure Estimation for Unmanned Aerial Vehicles \\against Adversarial Cyber Attacks}

\title{\Large \bf
	Secure Estimation for Unmanned Aerial Vehicles against Adversarial Cyber Attacks}

\author{Qie~Hu$^*$, Young Hwan Chang$^*$, and~Claire~J.~Tomlin
%\thanks{Qie Hu and Claire J. Tomlin are with the Department of Electrical Engineering and Computer Sciences, University of California, Berkeley, USA. \texttt{[qiehu, tomlin]@eecs.berkeley.edu}}
%\thanks{Young Hwan Chang is with the Department of Biomedical Engineering, Oregon Health and Science University, Portland, USA. \texttt{chanyo@ohsu.edu}}
%\thanks{*These authors contributed equally.}}
}

%\author{Qie~Hu$^{1,*}$, Young~Hwan~Chang$^{2,*}$, ~Claire~J.~Tomlin$^1$ }
%
%%\thanks{ Department of Electrical Engineering and Computer Sciences, University of California, Berkeley, CA 94720 USA (e-mail:$\{qiehu@eecs., yhchang@, tomlin@eecs.\}berkeley.edu$). }
%%\thanks{*These authors contributed equally}
%
%\author{\IEEEauthorblockN{Young Hwan Chang\IEEEauthorrefmark{1},
%Qie Hu\IEEEauthorrefmark{1}, and
%Claire J. Tomlin}
%\IEEEauthorblockA{\IEEEauthorrefmark{2}Department of Electrical Engineering and Computer Sciences,\\University of California, Berkeley, CA 94720 USA\\ Email: $\{yhchang@, 
%\IEEEauthorblockA{Department of Electrical Engineering and Computer Sciences,\\University of California, Berkeley, CA 94720 USA\\ Email: $\{yhchang@, qiehu@eecs., tomlin@eecs.\} berkeley.edu$.}
%\IEEEauthorblockA{\IEEEauthorrefmark{1}These authors contributed equally}}


%%%%%%%%%%%%%%%%%%%%%%%%%%%%%%%%%%%%%%%%
\maketitle
\thispagestyle{empty}
\pagestyle{empty}

% \textbf{Keywords.} list of keywords

\begin{abstract}
In the coming years, usage of Unmanned Aerial Vehicles (UAVs) is expected to grow tremendously. Maintaining security of UAVs under cyber attacks is an important yet challenging task, as these attacks are often erratic and difficult to predict. Secure estimation problems study how to estimate the states of a dynamical system from a set of noisy and maliciously corrupted sensor measurements. The fewer assumptions that an estimator makes about the attacker, the larger the set of attacks it can protect the system against.
In this paper, we focus on sensor attacks on UAVs and attempt to design a secure estimator for linear time-invariant systems based on as few assumptions about the attackers as possible. We propose a computationally efficient estimator that protects the system against arbitrary and unbounded attacks, where the set of attacked sensors can also change over time.
In addition, we propose to combine our secure estimator with a Kalman Filter for improved practical performance and demonstrate its effectiveness through simulations of two scenarios where an UAV is under adversarial cyber attack.

\end{abstract}

\vspace{3 mm}
\section{Introduction}
The already widespread use of Unmanned Aerial Vehicles (UAVs) is expected to continue to grow at a tremendous rate over the next few years \cite{faa}. 
Civilian applications of UAVs, such as cargo delivery \cite{Google, Amazon}, infrastructure surveillance \cite{Asctec}, and agricultural applications \cite{Croplife}, can provide great benefits to society. 

However, UAVs may be vulnerable to a variety of cyber attacks.
For example, to manage the increased UAV traffic, each UAV may periodically send its position measurements wirelessly to a remote traffic management center.
Similarly, two UAVs may exchange position and velocity information in a collaborative collision avoidance procedure.
These communication links could be subject to Man-In-The-Middle (MITM) attacks in which a malicious agent spoofs the information being sent and/or received. Successful attacks can lead to collisions of vehicles, economic loss and bodily damage.
Therefore, maintaining the security of UAVs under such cyber attacks is an important but also challenging task, as attacks are often erratic and difficult to model.

%On February 15, 2015, the Federal Aviation Administration proposed to allow routine use of certain small, non-recreational Unmanned Aerial Vehicles (UAVs) in today's aviation system \cite{faa}. Thus in the near future, we may see UAVs such as Amazon Prime Air \cite{Amazon} and Google Project Wing vehicles \cite{Google} sharing the airspace. In order to manage this UAV traffic, we may imagine a scenario in which each UAV periodically sends measurements such as its position and velocity wirelessly to a remote control center, which then estimates the vehicle's trajectory, for collision avoidance for example. The communications link between the UAV and the control center could be subject to Man-In-The-Middle (MITM) attacks in which a malicious agent spoofs the information being sent and/or received (Channel 1 in Figure \ref{fig:ex_uav_pic}) \cite{Welch}. Similar attack scenarios could arise in UAV formation: for formation control, individual UAVs receive information from other UAVs wirelessly in order to estimate other vehicles' positions (Channel 2 in Figure \ref{fig:ex_uav_pic}). 
%Maintaining security of UAVs under such cyber attacks is an important and challenging task, since these attacks can be erratic and thus difficult to model. 

Secure estimation problems study how to estimate the states of a dynamical system from a set of noisy and maliciously corrupted sensor measurements. 
In designing such estimators, it is desirable to make as few assumptions about the attackers as possible. This is because it is very difficult, if not impossible, to predict the behavior of attackers, and when an attack signal violates the assumptions of a secure estimator, then this estimator would fail to detect the attack.

Researchers have studied various approaches to securing general cyber-physical systems, each based on a different set of assumptions about the attacker.
For example, the authors in \cite{Bullo, Liu} assume that the attack signal would follow certain probabilistic distributions and then design filters for detection of such attacks. 
In \cite{Wu, Basar, Basar2, Walrand, Pappas}, the authors use the game theory framework, where the controller and attacker are players with competing goals in a game. Attackers are assumed to adopt specific strategies that maximize a certain cost and the controller or estimator is designed to minimize such a cost.
More recently, Fawzi \textit{et al.} proposed in \cite{Fawzi2014} a secure estimation method for arbitrary attacks, with a limiting assumption that the set of attacked sensors do not change with time. 

\pagestyle{empty}
In this paper, we focus on sensor attacks on UAVs and attempt to design a secure estimator for linear time-invariant (LTI) systems based on as few assumptions about the attackers as possible. First, we do not assume that the attack signals follow any stochastic distributions, and thus our proposed estimator works for arbitrary and unbounded attacks.
Second, we allow the set of attacked sensors to change over time.
The only assumption we make is that the number of attacked sensors is sparse.

We formulate this secure estimation problem into the classical error correction problem, from which we propose an $l_1$-optimization based estimator that is computationally efficient.
In addition, we prove the maximum number of sensor attacks that can be corrected with our estimator and propose a practical method for estimator design that guarantees accurate decoding.
Finally, to improve the estimator's practical performance, we propose to combine our secure estimator with a Kalman Filter (KF), and demonstrate its effectiveness using two examples of UAVs under adversarial cyber attacks.

%%%%%%%%%%%%%%%%%%%%%%%%%%%%%%%%%%%%%
%%%%%%%%%%%%%%%%%%%%%%%%%%%%%%%%%%%%%
\vspace{3mm}
\section{Classical Error Correction: A Review}\label{sec_review}
\subsection{Compressed Sensing}
Sparse solutions $x\in \mathbb{R}^n$, are sought to the following problem:
\begin{equation}
	\min_x \norm{x}_0 \text{ subject to } b= Ax
	\label{eq:CS}
\end{equation}
where $b \in \mathbb{R}^m$ are the measurements, and $A \in \mathbb{R}^{m\times n}~ (m \ll n)$ is a sensing matrix. $\norm{x}_0$ denotes the number of nonzero elements of $x$. The following lemma provides a sufficient condition for a unique solution to (\ref{eq:CS}).

%\textcolor{black}{in II.A., you reference [16] right before Lemma 1 and [18] inside Lemma 1.  This is confusing.  If you are referencing the form of Lemma 1 from [16] that is presented in [18] then you should be explicit about it.}
%\textcolor{cyan}{I removed [18] here, but we cite it later in IIB before equation (5).}
\begin{lem} (\hspace{1sp}\cite{Candes_Tao}) \label{lem:CS}
If the sparsest solution to (\ref{eq:CS}) has $\norm{x}_0 = q$ and $m\ge 2q$ and all subsets of $2q$ columns of $A$ are full rank, then the solution is unique.
\end{lem}
\begin{proof}
Suppose the solution is not unique. Therefore, there exists $x_1 \neq  x_2$ such that $Ax _1 = b$ and $Ax_2 = b$ where $\norm{x_1}_0 = \norm{x_2}_0 = q$. Then, $A(x_1 - x_2) = 0$ and $x_1 - x_2 \neq 0$. Since $\norm{x_1-x_2}_0 \leq 2q$ and all $2q$ columns of $A$ are full rank (i.e., linearly independent), it is impossible to have $x_1-x_2\neq 0$ that satisfies $A(x_1-x_2) = 0$. This contradicts the assumption.
\end{proof}

%%%%%%%%%%%%%%%%%%%%%%%%%%%%%%%%%%%%%
\subsection{The Error Correction Problem \cite{Candes_Tao}} %\label{sec:error_correction}
Consider the classical error correction problem: $y=Cx + e$ where $C\in \mathbb{R}^{l\times n}$ is a coding matrix $(l > n)$ and assumed to be full rank. We wish to recover the input vector $x \in \mathbb{R}^n$ from corrupted measurements $y$. Here, $e$ is an arbitrary and unknown sparse error vector. To reconstruct $x$, note that it is obviously sufficient to reconstruct the vector $e$ since knowledge of $Cx + e$ together with $e$ gives $Cx$, and consequently $x$ since $C$ has full rank \cite{Candes_Tao}. In \cite{Candes_Tao}, the authors construct a matrix $F$ which annihilates $C$ on the left, i.e.,  $FCx = 0$ for all $x$. Then, they apply $F$ to the output $y$ and obtain
\begin{equation}
	\tilde y = F (Cx + e) = Fe.
\end{equation}
Thus, the decoding problem can be reduced to that of reconstructing a sparse vector $e$ from the observations $\tilde y = Fe$. Therefore, by Lemma \ref{lem:CS}, if all subsets of $2q$ columns of $F$ are full rank, then we can reconstruct any $e$ such that $\| e \|_0 \leq q$.
%We refer to a decoder that can correct $q$ errors as a $q$-error-correcting decoder.

%%%%%%%%%%%%%%%%%%%%%%%%%%%%%%%%%%%%%
\vspace{3mm}
\section{Secure Estimation}

%\textcolor{black}{II.B. seems like it should reference [11] -- is this the exact formulation in [11] or is it changed to suit the development here?  In any case, be explicit about this.}

\subsection{Problem Formulation}
Consider the LTI system as follows:
\begin{equation}
\begin{aligned}
x(k+1) &= A_o x(k) + B u(k) \\
y(k) &= C x(k) + e(k),
\end{aligned} 
\label{eq:system_model_se}
\end{equation} 
where $x(k) \in \mathbb{R}^n$, $y(k) \in \mathbb{R}^p$ and $u(k) \in \mathbb{R}^m$ are the states, measurements and control inputs at time step $k$. $e(k) \in \mathbb{R}^p$ represents the attack signal at time $k$. Our goal is to reconstruct the initial state $x(0)$ of the plant from the corrupted observations $y(k)$'s where $k=0,...,T-1$.

The attack vector $e(k)$ is such that if the $i$th sensor is attacked at time $k$, then $e_i(k)$, the $i$th element of $e(k)$ is nonzero, otherwise $e_i(k) = 0$. We assume that the attack signal can be arbitrary and unbounded. In addition, we assume that the set of attacked sensors can change over time. As illustrated by the following example, if 2 sensors are attacked at each time step, we can have sensors 1 and 3 attacked at time step 0, sensors 2 and 3 attacked at time 1, and so on:
\begin{equation}
	\begin{bmatrix} e(0)  & \lvert & e(1) & \lvert &  ...  & \end{bmatrix} 
	= \begin{bmatrix} * & 0 & * & \cdots \\
					       0 & * & 0 & \cdots \\
					       * & * & 0 & \cdots \\
					       0 & 0 & * &\cdots 
			\end{bmatrix}, \nonumber \\
\end{equation}
where $*$ denotes a nonzero component (i.e., an attack or corruption). 

Furthermore, assume that a local control loop implements secure state feedback and is not subject to attack: $u(k) = Gx(k)$. In the case of UAVs, this corresponds to using measurements from onboard, hardwired sensors such as Inertial Measurement Units (IMU) for autopilot and trajectory following. The resulting closed loop system is:
\begin{equation}
\begin{aligned}
x(k+1) &= A x(k)\\
y(k) &= C x(k) + e(k),
\end{aligned} 
\label{eq:system_model_cl}
\end{equation} 
where the closed loop system matrix $A=A_o+BG$. 

Finally, we define the number of correctable attacks/errors as follows:
\begin{Def}\label{def:num_err_change}
When the set of attacked sensors/nodes can change over time, $q$ errors are correctable after $T$ steps by the estimator/decoder $\mathcal{D}: {(\mathbb{R} ^p) } ^T  \rightarrow \mathbb{R}^n$ if for any $x(0) \in \mathbb{R}^n$ and any sequence of vectors $e(0),...,e(T-1)$ in $\mathbb{R}^p$ such that $\lvert \textsf{supp}(e(k)) \rvert \leq q$, 
we have $\mathcal{D} (y(0),...,y(T-1)) = x(0)$ where $y(k) = CA^k x(0) + e(k)$ for $k=0,...,T-1$.
\end{Def}

%In practice, this represents the following scenario: a physical system possesses a local control loop that has direct access to the state of the plant and can control the evolution of the physical system. This is reasonable if the sensors are connected to the local controller through a wired link that is not subject to external attacks. Also, as part of the overall plant, a higher-level supervisory and monitoring system receives measurements from the sensors through wireless and vulnerable communication links that are subject to attacks \cite{Fawzi2014}. 
%A concrete example is a UAV that uses measurements from onboard, hardwired sensors such an Inertial Measurement Unit (IMU) for its autopilot and trajectory following (i.e. secure local control loop), and communicates wirelessly with a remote control center (i.e. vulnerable link subject to attack).

%%%%%%%%%%%%%%%%%%%%%%%%%%%%%%%%%%%%

\subsection{Methodology}

Let $E_{q,T}$ denote the set of error vectors $\begin{bmatrix} e(0); ~ ...~  ;  e(T-1) \end{bmatrix}   \in  \mathbb{R}^{p\cdot T} $ where each $e(k)$ satisfies $\|e(k)\|_0 \leq q \leq p$. %$\lvert \textsf{supp}(e(k)) \rvert \le q \le p $.
\begin{eqnarray} \label{eq:sys_err_corr}
\begin{aligned}
	Y &\triangleq \begin{bmatrix} y(0) \\ y(1) \\ \vdots \\ y(T-1) \end{bmatrix}  = \begin{bmatrix} Cx(0) + e(0)\\ CA x(0) + e(1) \\ \vdots \\ CA^{T-1} x(0) + e(T-1) \end{bmatrix} \\
		& =
		\begin{bmatrix} C \\ CA \\ \vdots \\ CA^{T-1} \end{bmatrix} x(0) + E_{q,T}  \triangleq \Phi x(0) + E_{q,T}
		\label{eq:decoder_Phi}
\end{aligned}
\end{eqnarray}
where $Y \in \mathbb{R}^{p\cdot T}$ is a collection of corrupted measurements over $T$ time steps and $\Phi \in \mathbb{R}^{p\cdot T \times n}$ represents an observability-like matrix of the system. %\st{ if $B=0$ (i.e., $A = A_o$)}. \st{Since we consider secure estimation of adversarial attacks,}
Here, we need to assume that $\operatorname{rank}(\Phi) = n$; otherwise, the system is unobservable and we cannot determine $x(0)$ even if there is no attack (i.e., $E_{q,T} = 0$).

Inspired by the error correction techniques proposed in \cite{Candes_Tao} and \cite{David_Chang}, we first determine the error vector $E_{q,T}$, and then solve for $x(0)$. %[We say this upfront, so readers know what to expect, and it may be less surprising when they see the second method?]  \st{We consider the error correction approach.}
Consider the $QR$ decomposition of $\Phi \in \mathbb{R}^{p\cdot T \times n}$,
\begin{eqnarray}
	\Phi = \begin{bmatrix} Q_1 & Q_2 \end{bmatrix} \begin{bmatrix} R_1 \\ 0 \end{bmatrix} = Q_1 R_1
\end{eqnarray}
where $\begin{bmatrix} Q_1 & Q_2 \end{bmatrix} \in \mathbb{R}^{p\cdot T \times p\cdot T}$ is orthogonal, $Q_1 \in \mathbb{R}^{p\cdot T\times n}, Q_2 \in \mathbb{R}^{p\cdot T \times (p\cdot T-n)}$, and $R_1 \in \mathbb{R}^{n\times n}$ is a rank-$n$ upper triangular matrix.
Pre-multiplying (\ref{eq:decoder_Phi}) by $\begin{bmatrix} Q_1 & Q_2 \end{bmatrix} ^\top$ gives:
\begin{equation}
	\begin{bmatrix} Q_1 ^\top \\ Q_2 ^\top \end{bmatrix} Y = \begin{bmatrix}R_1 \\ 0  \end{bmatrix} x(0) + \begin{bmatrix} Q_1 ^\top \\ Q_2^\top \end{bmatrix} E_{q,T}.
	\label{eq:QR}
\end{equation}
We can compute $E_{q,T}$ by using the second block row:
\begin{equation}
	\tilde Y \triangleq Q_2^\top Y = Q_2^\top E_{q,T}
	\label{eq:E_est}
\end{equation}
where $Q_2^\top \in \mathbb {R} ^{ (p\cdot T-n) \times p\cdot T}$.
From Lemma \ref{lem:CS}, (\ref{eq:E_est}) has a unique, $s$-sparse solution (where $s\le q\cdot T$) if all subsets of $2s$ columns (at most $2 q\cdot T$ columns) of $Q_2^\top$ are full rank. Clearly, this is a reasonable assumption if $(p\cdot T-n) \ge 2q\cdot T$. Therefore, we consider solving the following $l_1$-minimization problem:
\begin{equation}
	\hat{E}_{q,T} = \arg \min_E \norm { E}_{l_1} \text{ s.t. } \tilde Y = Q_2^\top E
	\label{eq:solve_E}
\end{equation}
Now, given the vector $\hat{E}_{q,T}$, we can compute $x(0)$ from the first block row of (\ref{eq:QR}) as follows:
\begin{equation}
	x(0) = R_1^{-1} Q_1^\top (Y- \hat{E}_{q,T})
	\label{eq:QR1}
\end{equation}
The following lemma provides the conditions under which the solution to (\ref{eq:QR1}) exists and is unique.
%\textcolor{black}{Specifically, does Lemma 2 need a citation?} \textcolor{green}{I think we need a citation here}\textcolor{cyan}{This is a result proved in [11], is it OK to just cite [11] or do we need any other citations?}
%\begin{lem} \label{lem:EC}
%	$x(0) $ is the unique solution if all subsets of $2s$ columns of $Q_2 ^\top$ are linearly independent and $\Phi$ is full column rank. %Also, this condition is equivalent to $\| \Phi z \|_0 > 2s$ %$\lvert \textsf{supp}( \Phi z) \rvert > 2 s = 2 (q\cdot T)$
%%for all $z \in \mathbb{R}^n \backslash \{ 0 \}$.
%\end{lem}
\begin{lem} \label{lem:EC}
	$x(0) $ is the unique solution if $\lvert \textsf{supp}( \Phi z) \rvert > 2 s = 2 (q\cdot T)$ for all $z \in \mathbb{R}^n \backslash \{ 0 \}$.
\end{lem}
\begin{proof}
We first prove the claim C1: if $\lvert \textsf{supp}( \Phi z) \rvert > 2 s = 2 (q\cdot T)$ for all $z \in \mathbb{R}^n \backslash \{ 0 \}$ then all subsets of $2s$ columns of $Q_2^\top$ are full rank. Then by Lemma \ref{lem:CS} and noting that by definition the null space of $Q_2^\top$ equals the column space of $\Phi$, we have $x(0) $ is the unique solution.

Proof of C1 by contradiction: Suppose there exist $2s$ columns of $Q_2^\top$ that are linearly dependent. Then, there exists $E_0 \neq 0$ such that $Q_2^\top E_0 = 0$ where $\lvert \textsf{supp}(E_0) \rvert \le 2s$. Since the null space of $Q_2^\top$ equals the column space of $\Phi$, there exists $z$ such that $E_0 = \Phi z$ (i.e., $E_0$ is in the column space of $\Phi$). Then, $ \lvert  \textsf{supp}(\Phi z) \rvert = \lvert \textsf {supp} (E_0) \rvert \le 2s $ (contradiction).
\end{proof}

The sufficient condition, provided in Lemma \ref{lem:EC}, for the existence of a unique solution to (\ref{eq:QR1}) is hard to check as it requires satisfiability of the condition for all $z \in \mathbb{R}^n \backslash \{ 0 \}$. In the following Theorem, we prove an equivalent, yet simple-to-check, sufficient condition that only needs to be verified for the eigenvectors of $A$.

%The significance of this lemma is that in order to check whether a decoder can guarantee accurate decoding of $q$ errors when the attacked nodes are fixed, one no longer needs to check satisfiability of condition $(i)$ which is stated for all $z \in \mathbb{R}^n \backslash \{0\}$ and hard to check, instead, one can simply check condition $(ii)$ for the eigenvectors of $A$ which is much simpler. Next, we derive a similar result for our decoder for when the attacked nodes can change with time.

\begin{thm} Let $A \in \mathbb{R}^{n\times n}, C \in \mathbb{R}^{p\times n}$. Assume that $C$ is full rank, $(A,C)$ is observable and $A$ has $n$ distinct positive eigenvalues such that $0 < \lambda_1 < \lambda_2 < \cdots < \lambda_n$. %, \st{and each row of $C$ is not identically zero and is not redundant } . 
Define:
\begin{itemize}
\item
$s_i \triangleq \lvert \textsf{supp} (Cv_i) \vert$, where $v_i$ is an eigenvector of $A$, %($Av_i =\lambda_i v_i$),
\item
$\mathcal{S} \triangleq \{ s_1, s_2, \cdots, s_n \}$,
\item
For every $m \in \{2, \ldots, n\}$, let $\mathcal{S}_m$ be any subset of $\mathcal{S}$ with $m$ elements, define $T_{\mathcal{S}_m} \triangleq \frac {  (m-2) \cdot p + \min \mathcal{S}_m } {\max \mathcal{S}_m - 2q }$.
Then $T_m$ is such that $T_m > T_{\mathcal{S}_m}$ for all subsets $\mathcal{S}_m$, i.e. all subsets of $m$ elements from the set $\mathcal{S}$.
\end{itemize}
Choose $T$ such that  $T \ge \max \{ T_2, \cdots, T_n \}$.
Then, the following are equivalent:
%(for $m=1$, we have shown that (i) and (ii) in Equation (14) are equivalent where $T=1$):
\begin{equation}
\begin{aligned} 
 (i)  &~\forall v_i \in \mathbb{R}^n \text{ where } Av_i =\lambda_i v_i, \\
 	&~ \lvert \textsf{supp}(Cv_i) \rvert > 2q  \\
  (ii)  &~\forall v_i \in \mathbb{R}^n \text{ where } Av_i =\lambda_i v_i,  \\
  	&~\lvert \textsf{supp} (\Phi v_i) \rvert > 2q \cdot T  \\
  (iii) &~  \forall z \in \mathbb{R}^n\backslash \{0 \}, \lvert \textsf{supp} (\Phi z) \rvert > 2 q \cdot T\\
  \nonumber 
\label{eq:new_condition}
\end{aligned}
\end{equation}
\end{thm}
\noindent

%In order to prove Theorem 1, we make use of Lemmas \ref{lem:two_vec}, \ref{lem:three_vec} and Proposition \ref{prop:m_vec} (see Appendix): 

\begin{proof} Interested readers are referred to the proof for Theorem 1 in our archived paper \cite{Chang_Qie}.

\end{proof}

%\begin{lem} \label{lem:controllability}
%Assume that the pair $(A_o,B)$ is controllable. Then the closed-loop system with state feedback is controllable and thus, there exists $G $ such that the eigenvalues of the closed-loop matrix $A$ $(=A_o+BG)$, i.e., $\lambda_1, ..., \lambda_n$ can be arbitrarily located on the complex plane.
%\end{lem}

%\begin{prop}\label{prop:equivalent2}
%Given $C$ is full rank, the closed-loop matrix $A~(=A_o+BG)$ has $n$ distinct positive eigenvalues, the open-loop pair $(A_o,B)$ is controllable, the closed-loop pair $(A,C)$ is observable and $T$ is chosen to satisfy Theorem 1. Then, the condition for secure estimation of $q$-errors when the set of attacked nodes is fixed ((i) in (\ref{eq:connection})) is the same as the condition for when the set of attacked nodes can change over time ((ii) in (\ref{eq:connection})), except the condition on $T$.
%\end{prop}
%\begin{proof}
%Since the pair $(A_o,B)$ is controllable, there exists a feedback matrix $G$ such that the eigenvalues of the closed-loop matrix $A$, i.e., $\lambda_1, ..., \lambda_n$ can be arbitrarily located on the complex plane. Then Proposition \ref{prop:equivalent2} directly follows from Proposition \ref{prop:equivalent}, Lemmas \ref{lem:distinct} and Theorem 1 
%\end{proof}

\noindent
Theorem 1 %and Proposition \ref{prop:equivalent2} state 
states that if the feedback system and the secure estimator are designed such that all the conditions in the theorem are satisfied, then our proposed secure estimator can guarantee accurate correction of $q$ errors by checking the following very simple condition:
\begin{equation}
\forall v_i \in \mathbb{R}^n \text{ where } Av_i =\lambda_i v_i, ~ \lvert \textsf{supp}(Cv_i) \rvert > 2q.  \nonumber
\end{equation}
%And interestingly, this is the exact same condition that one should check if one is designing the decoder from \cite{Fawzi2014} for fixed attack nodes. In other words, it is equally easy to check satisfiability of the sufficient condition for $q$-error-correction for both types of decoders.

%%%%%%%%%%%%%%%%%%%%%%%%%%%%%%%%%%%%%%
\vspace{1mm}
\subsection{Number of Correctable Errors}\label{sec:max_q}

Given that the set of attacked nodes can change over time and $e(k)$ satisfies $\lvert \textsf{supp} (e(k)) \rvert \le q$ for all $k$, we prove in Proposition \ref{prop:maximum} (see below) that the maximum number of correctable errors (as defined in Definition \ref{def:num_err_change}) by our decoder is $\lceil p/2-1 \rceil$, where $p$ is the number of measurements. 
%This is in fact the same as the maximum number of correctable errors for the decoder proposed in \cite{Fawzi2014} which is for fixed attacked nodes.
%This is a pleasing result, because it demonstrates that with our proposed decoder, we can relax the assumption of fixed attacked nodes and protect the system against more general attacks, without compromising the maximum number of correctable errors. 
\begin{prop}\label{prop:maximum} 
Let $A_0 \in \mathbb{R}^{n \times n}$, $B \in \mathbb{R}^{n \times m}$ and $C \in \mathbb{R}^{p \times n}$ and assume that the pair ($A_0$, $B$) is controllable, $C$ is full rank and each row of $C$ is not identically zero. Then there exists a finite set $F \subset \mathbb{R}_+$ such that for any choice of $n$ numbers $\lambda_1, \cdots, \lambda_n \in \mathbb{R}_+ \backslash F$ such that $0<\lambda_1< \cdots < \lambda_n$, there exists $G \in \mathbb{R}^{m \times n}$ such that:
\begin{itemize}
\item
The eigenvalues of the closed-loop matrix $A~(= A_0+BG)$ are $\lambda_1, \cdots, \lambda_n$.
\item
If the pair ($A, C$) is observable, then the number of correctable errors for the pair ($A, C$) is maximal after $T= \max\{n, T^*\}$ time steps and is equal to $\lceil p/2-1 \rceil$, where $T^*$ is the value of $T$ from Theorem 1. 
\end{itemize}
\end{prop}
\begin{proof}
The proof for Proposition 4 in \cite{Fawzi2014} shows that if the chosen poles $\lambda_1, \cdots, \lambda_n$ are distinct, positive and do not fall in some finite set $F$, then there is a choice of $G$ such that the eigenvalues of $A~(=A_0+B)$ are exactly $\lambda_1, \cdots, \lambda_n$, and the corresponding eigenvectors $v_i$ are such that $\lvert \textsf{supp} (C v_i) \rvert = p$. Thus, by Theorem 1, the number of correctable errors for $(A,C)$ is $\lceil p/2-1 \rceil$.
\end{proof}

In addition, recall that $E_{q,T}$ consists of the error vectors $e(0), \cdots, e(T-1)$ stacked vertically and our proofs for the existence of a unique solution to (\ref{eq:QR1}) are independent of how the individual error (nonzero) terms are distributed in the vector $E_{q,T}$. Thus, we can remove the assumption: $\lvert \textsf{supp} (e(k)) \rvert \le q$ for all $k$, and allow $e(k)$ to appear in an arbitrary fashion, e.g. $\lvert \textsf{supp} (e(0)) \rvert = 2q$ and $\lvert \textsf{supp} (e(1)) \rvert = 0$, as long as $\sum_{k=0}^{T-1} \lvert \textsf{supp} (e(k)) \rvert \leq q\cdot T$, then our $q$-error-correcting decoder can still recover the true states. In other words, our proposed secure estimator can protect the system against more general attacks where the number of attacked sensors is not necessarily less than or equal to $q$ at every time step.

\vspace{3mm}
\section{Combination of Secure Estimation and Kalman Filter}\label{sec:estimation}
Consider the state estimation problem for the following LTI system under attack:
\begin{equation}
\begin{aligned}
x(k+1) & = A x(k) + B u(k)\\
y(k) & = C x(k) + e(k) + v(k),
\end{aligned}
\end{equation}
where $x$, $y$, $u$ and $e$ are as defined in (\ref{eq:system_model_se});
%$x$ is the state, $u$ is the control input, $y$ is the output, $e$ is the attack signal, 
and $v$ is a zero mean independent and identically distributed (i.i.d.) Gaussian measurement noise. 

A KF can be used to estimate the states by modeling the attack signal as noise. More specifically, define a new measurement noise $\bar{v}(k) = e(k) + v(k)$ to give a new measurement equation $y(k) = C x(k) + \bar{v}(k)$. A KF can then estimate the states from the inputs $u(k)$ and the corrupted measurements $y(k)$ \cite{KwonACC}. 
One caveat with this method is that KFs assume zero mean and i.i.d. white Gaussian measurement noise, however, attack signals are usually erratic and may be poorly modeled by Gaussian processes \cite{KwonACC}, i.e., $e(k)$ and consequently, $\bar v(k)$ may not be Gaussian. Take Global Positioning System (GPS) spoofing attacks for example, attack signals are often structured to resemble normal GPS signals or can be genuine GPS signals captured elsewhere. %Such signals are neither Gaussian nor white. 
When the system is subjected to attacks that are poorly modeled by Gaussian processes, it is reasonable to expect KFs to fail to recover the true states. 
%Figure \ref{fig:estimation} gives an illustrative example where an attack signal that increases linearly with time is injected into the measurements of state $x_i$. The red dashed line shows a plausible estimated state trajectory from a KF.

On the other hand, our proposed secure estimator does not assume the attack signal to follow any model, and therefore, it works for arbitrary and unbounded attacks. The only assumption is that the number of attacked sensors is sparse, i.e., less than $\lceil p/2-1\rceil$. 
As the set of attacked sensors becomes less sparse, our secure estimator occasionally fails to recover the true states. 
Based on these observations, we propose to combine our secure estimator with a KF to improve its practical performance, as detailed in Algorithm 1.
\begin{algorithm}
\caption{Combined secure estimator with KF}
\label{al:se_kf}
\begin{algorithmic}[1]
\State Initialize the KF
\For{each $k$}
	\If{$k \geq T$}
		\State Estimate the attack signal at time $k$, $\hat e(k)$, using secure estimator
	\Else
		\State Set $\hat e(k) = 0$
	\EndIf
	\State Form a new measurement equation: $\tilde y(k) =  C x(k) + \tilde v(k)$, where $\tilde y(k) = y(k) - \hat e(k)$ and $ \tilde v(k) = e (k) - \hat e(k) + v(k)$
	\State Apply standard KF using $u$ and $\tilde y$ 
\EndFor
\end{algorithmic}
\end{algorithm}

The intuition is that the secure estimator acts as a pre-filter for the KF, so that $\tilde v(k)$ is close to a zero mean i.i.d. Gaussian process even when the true attack signal $e(k)$ is not. 
At most time steps $k$, the secure estimator perfectly recovers $e(k)$, i.e., $\hat e(k) = \hat e(k)$, hence $\tilde v(k) = v(k)$ and thus, is a zero mean Gaussian process. What happens when the secure estimator fails? (\ref{eq:sys_err_corr}) shows that the estimated state at time $k$, $\hat x(k)$, is independent from the estimated state at another time step $\hat x(l)$ ($l \neq k$). 
As a result, when the secure estimator fails, its estimation error, $e(k) - \hat e(k)$, appears to be quite random. 
Putting these together: $\tilde v(k) = e(k) - \hat e(k) + v(k)$ is closer to a zero mean i.i.d. white Gaussian process than $\bar v(k)$ (i.e., the corresponding measurement noise if a KF is applied directly to estimate the states), which improves the KF's performance. 
Finally, the \textit{if} statement in Algorithm 1 ensures that the secure estimator always has access to $T$ past measurements.

Next, we demonstrate the effectiveness of our proposed method through simulations of a UAV under two types of adversarial attacks, which also provides a realistic example illustrating the behaviors described in this section.

%%%%%%%%%%%%%%%%%%%%%%%%%%%%%%%%%%%%%
%%%%%%%%%%%%%%%%%%%%%%%%%%%%%%%%%%%%%
\vspace{3mm}
\section{Numerical Examples}

\subsection{UAV Model}
%We consider a 10-D quadrotor model with the following dynamics:
We consider a quadrotor with the following dynamics:
\begin{equation}
\begin{aligned}
x(k+1) &= A_0 x(k) + B u(k)  + g\\
y(k) &= C x(k) + e(k) + v(k), \\
\end{aligned}
\end{equation}
\noindent 
where $x = [p_x, v_x, \theta_x, \dot \theta_x, p_y, v_y, \theta_y, \dot\theta_y, p_z, v_z]^T$ is the state vector. $p_x$, $p_y$ and $p_z$ represent the quadrotor's position along the $x$, $y$ and $z$ axis, respectively, and $v_i$'s are their corresponding velocities. $\theta_x$ and $\theta_y$ are the pitch and roll angles respectively, and $\dot \theta_i$'s are their corresponding angular velocities. %\st{represents the quadrotor's position, velocity, pitch and roll angles and angular velocities.} 
The input vector $u = [\theta_{r,x}, \theta_{r,y}, F]^T$, where $\theta_{r,i}$ is the reference pitch or roll angle, and $F$ is the commanded thrust in the vertical direction. $y = [\tilde{p}_x, \tilde{p}_y, \tilde{p}_z]^T$ %\st{$y = [p_x, p_y, p_z]^T$}
represents corrupted position measurements under attack $e$ and measurement noise $v$.
The constant vector $g$ represents gravitational effects and can be dropped without loss of generality because we can always subtract it out in $u$. 
Further details about this model and its derivation can be found in \cite{Bouffard}. Finally, the matrix $C$ depends on the particular measurements taken in each example.

%%%%%%%%%%%%%%%%%%%%%%%%%%%%%%%%%%%%%%%%%%%%%%%%

\subsection{Decoder Design via Pole-Placement}

We assume that the UAV uses the state feedback control law $u(k) = G x(k)$\footnote{In the GPS spoofing example, direct uncorrupted state measurements are not available. Therefore a KF is used to give estimated states which are then used for state feedback control.}, where $G$ is the feedback matrix which can be designed. In this section, we show that we can design $G$ to achieve our desired trade-off between the control performance and the secure estimation performance. 

If the open loop pair $(A_0, B)$ is controllable, then the closed loop poles can be placed anywhere in the complex plane by appropriate choice of $G$. First, we design a Linear Quadratic Regulator (LQR) and evaluate its secure estimation performance: we check the number of errors that the resulting secure estimator can correct by finding the maximum $q$ for which $\lvert \textsf{supp} (C v_i) \rvert > q$ for all $i$. 
Figure \ref{fig:ex_pole} shows the results for a matrix $C \in \mathbb{R}^{5\times 10}$ (i.e., 5 measurements). Observe that $\lvert \textsf{supp} (C v_i) \rvert < p = 5$ for $i=1,2,9,10$ and furthermore, $\lvert \textsf{supp} (C v_i) \rvert = 1 > 0$ for $i = 9$ and $10$, which means that the resulting secure decoder can correct zero errors!
As shown in Figure \ref{fig:ex_pole}, to improve the secure estimation performance, we perturb the closed-loop poles slightly until $\lvert \textsf{supp} (C v_i) \rvert = p$ for all $i$, i.e., we design a secure decoder that can achieve the maximum number of correctable errors within the limits of $p$ (i.e., the number of measurements). By keeping the perturbations on the poles small, our final controller achieves both good control and estimation performances.% (see Figure \ref{fig:ex_pp_est}).% and \ref{fig:ex_pp_err}). 

%%%%%%%%%%%%%%%%%%%%%%%%%%%%%%%%%%
\begin{figure}[t]
\center
\includegraphics[width=0.40\textwidth]{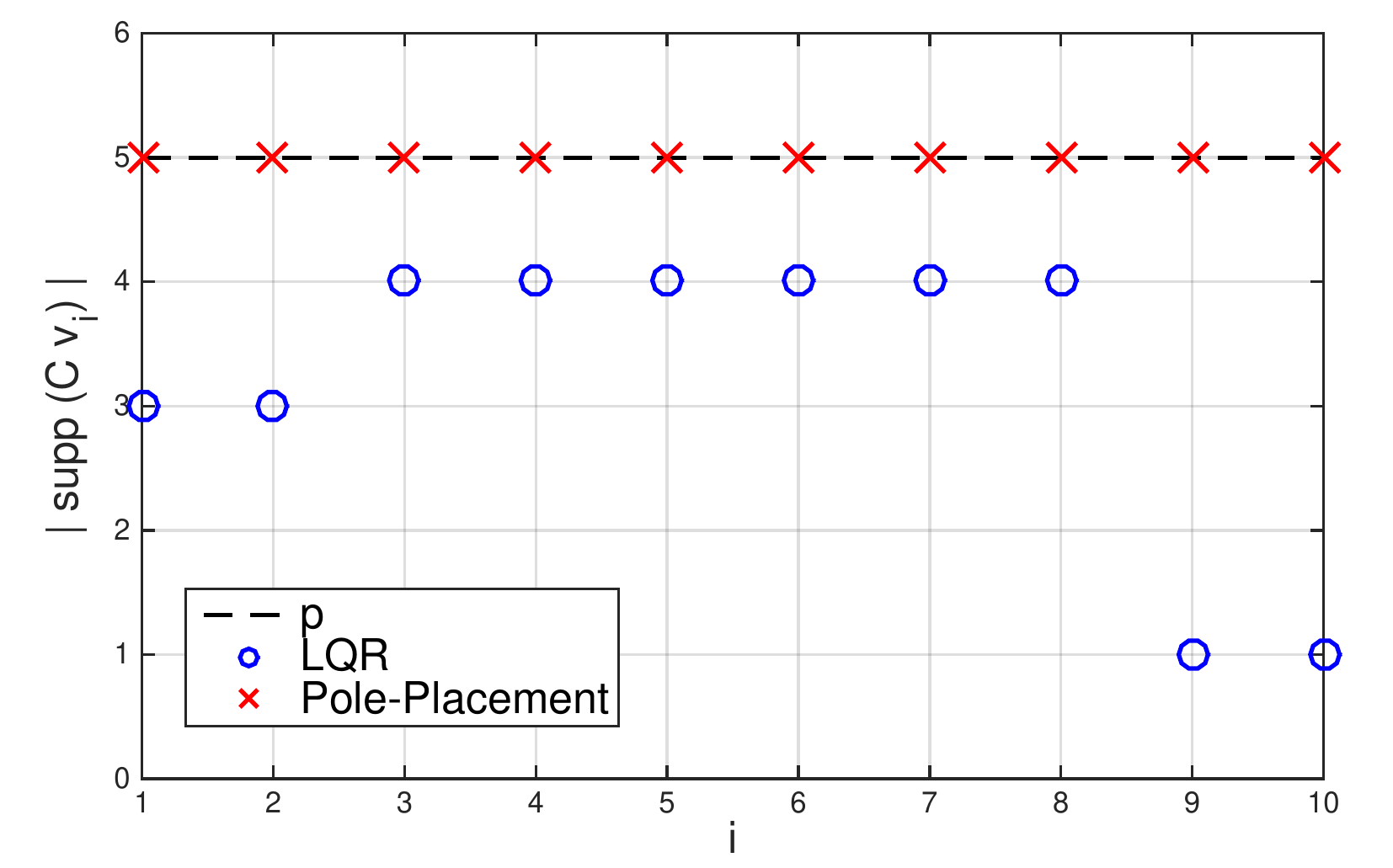}
\caption{$ \lvert \textsf{supp} (C v_i) \rvert $ for all eigenvectors $v_i$ of the closed-loop matrix $A$ for 2 feedback controllers: a LQR and a controller designed by pole-placement. Black dashed line is at $p = 5$, i.e., the number of measurements.}
\label{fig:ex_pole}
\end{figure}

%%%%%%%%%%%%%%%%%%%%%%%%%%%%%%%%%%

%%%%%%%%%%%%%%%%%%%%%%%%%%%%%%%%%%%%%
\begin{figure}
\center
\includegraphics[width=0.35\textwidth]{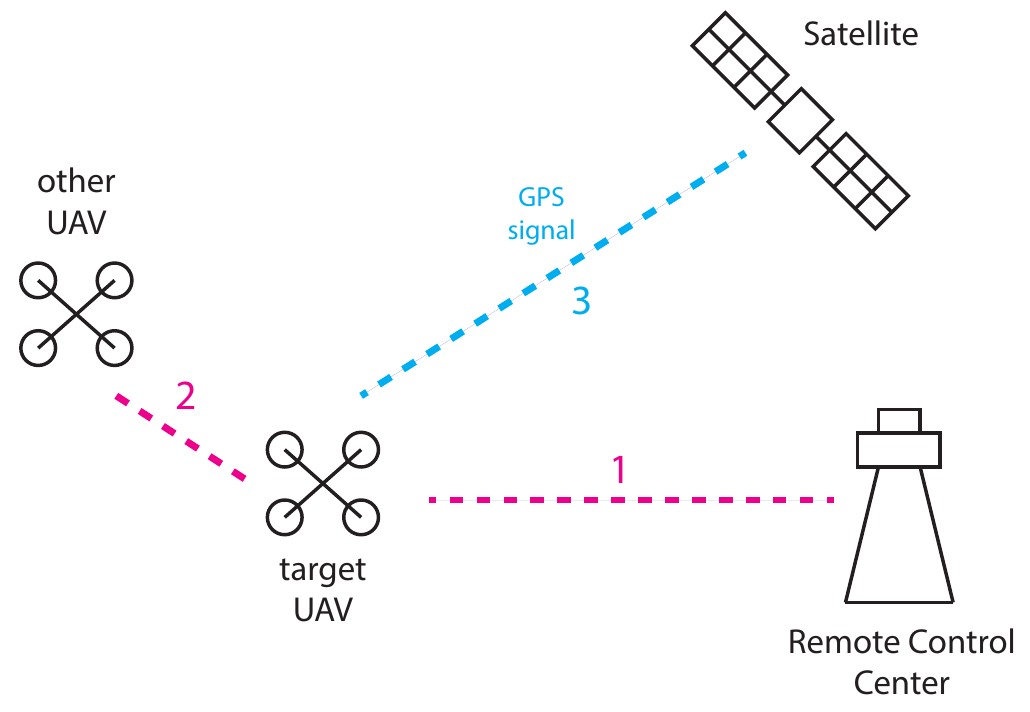}
\caption{Different communication channels that are subjected to cyber attacks.}
\label{fig:ex_uav_pic}
\end{figure}
%%%%%%%%%%%%%%%%%%%%%%%%%%%%%%%%%%%%%

%%%%%%%%%%%%%%%%%%%%%%%%%%%%%%%%%%%%%%%%%%%%%%%%

\subsection{UAV under Adversarial Cyber Attack}

\subsubsection{MITM Attack in Communication} \label{sec:uav_utm}
In this section, we consider MITM attacks targeted at Channels 1 and 2 in Figure \ref{fig:ex_uav_pic}, where a malicious agent spoofs the information being sent and/or received over these channels. 
The goal of the remote control center or the other UAV is to accurately estimate the true flight path of the target UAV from compromised measurements. 
Note that the attack does not affect the actual path of the target UAV (as opposed to the GPS spoofing example later in this section).

Assume that the attacker aims to deceive the receiver that the target UAV is deviating in the $x$-direction, therefore she spoofs the $x$-position measurements by injecting a continuous and increasing attack signal in $p_x$.
To make the estimation task even harder for the receiver, at each time step, the attacker injects a Gaussian noise to an additional randomly selected measurement, and the choice of this measurement changes over time. 

In this example, we first demonstrate the effectiveness of our proposed decoder design using the pole-placement method by comparing the estimation performance of the decoder resulting from (1) a LQR controller and (2) a controller designed using pole-placement as described in the previous section.
We then implement the latter controller, and compare the performance of three different estimation schemes: (1) KF only (KF), (2) secure estimator only (SE), and (3) secure estimator combined with KF (KF+SE). 

%%%%%%%%%%%%%%%%%%%%%%%%%%%%%%%%%%%%%
\begin{figure}
\center
\includegraphics[width=0.48\textwidth]{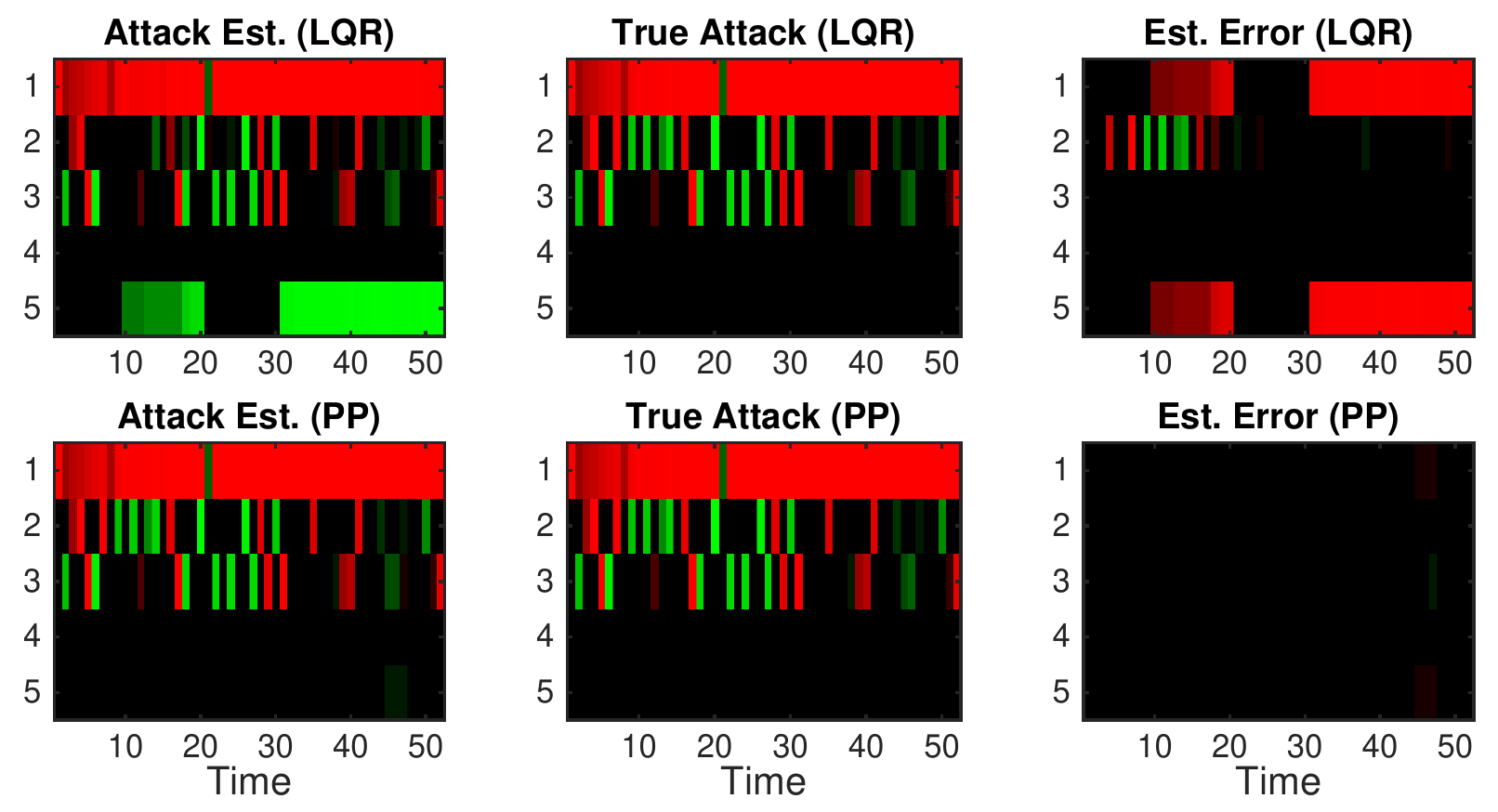}
\caption{Estimated attack signal, true attack signal and estimation error in the attack signal of the estimator (SE) with 2 different feedback controllers: LQR, controller designed via pole-placement (PP); with 5 measurements. Left column shows estimated attack signals. Middle column shows true attack signal. Right column shows estimation error. Each row corresponds to one type of measurement. Red pixels indicate positive values, green pixels are negative values and black indicates zero. }
%\qie{Is it necessary to include this figure? It seems repetitive from previous one.} \yh{It is repetitive from the previous one but it may help to see the structure of $E(k)$, i.e., attacked nodes change over time. Also, it would be good to understand Fig. 9}}
\label{fig:ex_pp_err}
\end{figure}
%%%%%%%%%%%%%%%%%%%%%%%%%%%%%%%%%%%%%

Throughout this example, $y \in \mathbb{R}^5$ and the measurements include the $x$, $y$ and $z$ positions and two additional randomly selected states. 
Figure \ref{fig:ex_pp_err} compares the accuracy of the estimated attack signals by the LQ regulator (top) and the one designed via pole-placement (bottom).
In each plot, one row corresponds to one sensor, and the first 3 rows are the $x$, $y$ and $z$ position measurements, respectively. The color of the pixel indicates the value of the signal or the estimation error.
The middle plots show the true attack signal and they highlights three points: first, the attacked sensors change with time; second, the number of attacked sensors at each time step $k$ is less or equal to 2; third, only position measurements are corrupted.
The left plots show the estimated attacked signal by each decoder. It is easy to see that the decoder resulting from a feedback controller designed via pole-placement estimates the attack signal much more accurately. 
The right plots of this figure highlight this observation by explicitly showing the estimation error of the attack signal for each measurement.

Figure \ref{fig:ex_uav_remote} compares the estimated flight paths by all three methods: KF, SE and KF+SE.
The UAV starts from the blue triangle and follows the solid blue line to land at the blue square. The estimated paths by each method are shown in red dashed lines. Observe that the KF fails to filter out the attack signal in the $x$-position measurements as the attack is highly non-Gaussian, and the estimated trajectory differs significantly from the true one. 
On the other hand, SE correctly estimates most portions of the trajectory and the final position of the vehicle, nevertheless it gives spontaneous errors. 
Finally the combined method KF+SE perfectly recovers the true path of the target UAV.

%%%%%%%%%%%%%%%%%%%%%%%%%%%%%%%%%%%%%
\begin{figure}
\center
\includegraphics[width=0.5\textwidth]{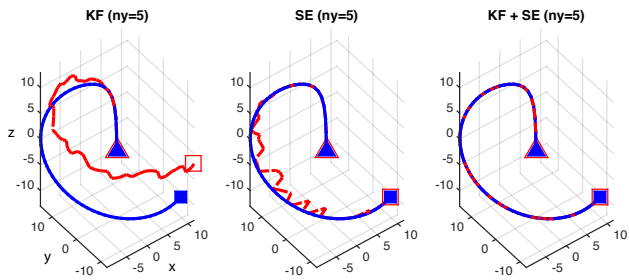}
\caption{Estimated UAV trajectory by three methods under MITM attack: KF only (KF), secure estimator only (SE), secure estimator with KF (KF+SE). Solid blue lines are the true UAV trajectories. They start from the blue triangle and end at the blue square. Red dashed lines represent estimated trajectories by each method, with 5 measurements.}
\label{fig:ex_uav_remote}
\end{figure}
%%%%%%%%%%%%%%%%%%%%%%%%%%%%%%%%%%%%%

%%%%%%%%%%%%%%%%%%%%%%%%%%%%%%%%%%%%%
\begin{figure}
\center
\includegraphics[width=0.5\textwidth]{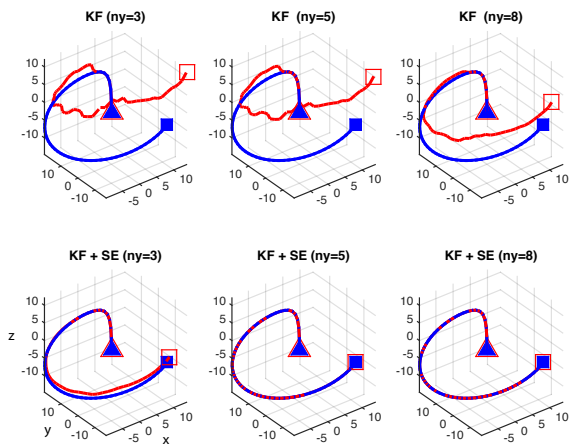}
\caption{Desired and actual UAV trajectory in different cases: KF and KF+SE, each using 3, 5 and 8 different measurements. Blue solid lines are the desired trajectory. Red dash lines are the actual UAV trajectory under adversarial attack.}
\label{fig:ex_uav_traj}
\end{figure}
%%%%%%%%%%%%%%%%%%%%%%%%%%%%%%%%%%%%%

\vspace{1mm}
\subsubsection{GPS Spoofing}

In this section, we focus on adversarial attacks in the GPS navigation system (Channel 3 in Figure \ref{fig:ex_uav_pic}). Consider the scenario where a UAV uses a Linear Quadratic Gaussian (LQG) controller to follow a desired trajectory, $x_r(k)$. In other words, a KF uses corrupted and noisy measurements $y(k)$ to produce a state estimate $\hat x(k)$, which is then used for state feedback control: $u(k) = G (\hat x(k) - x_r(k))$, where $G$ is the feedback matrix. 
Note that in the previous example (Section \ref{sec:uav_utm}), the feedback controller had access to uncompromised state measurements $x(k)$, therefore the true trajectory of the UAV is unaffected by attacks. 
In this example, however, the UAV uses estimated states $\hat x(k)$ for feedback control and path following. Therefore, if the measurements are corrupted and the state estimates are poor, then the UAV may deviate away from its desired path. Hence, the goal of the UAV is to accurately follow its planned trajectory in the presence of cyber attacks. 

Assume an attacker spoofs the GPS position measurements in order to deviate the UAV from its desired path. She injects a sinusoidal signal into the $x$ position measurement, as well as a Gaussian noise to a randomly chosen position measurement at each time step. 

In this example, we explore the effect of the number of sensor measurements on the secure estimation performance of two schemes: (a) KF only, (b) KF+SE.
First, we assume that the UAV only uses GPS for navigation, i.e., 3 positional measurements. 
Figure \ref{fig:ex_uav_traj} shows that KF completely fails to estimate the attack signal (KF, $n_y = 3$), consequently, the actual UAV trajectory (red dashed line) deviates significantly from its desired path (solid blue line).
On the other hand, the combined method KF+SE's estimated attack signals are significantly more accurate, therefore the UAV can follow its planned path much more closely (Figures \ref{fig:ex_uav_traj}, KF + SE, $n_y = 3$).
Recall from Proposition \ref{prop:maximum} that the maximum number of correctable errors for a system with $p$ measurements is $\lceil p/2-1 \rceil$, which equals 1 in this case. However, at any time step $k$, there are at most 2 attacked sensors, which exceeds the above limit and explains the estimation error of the combined method KF+SE. Despite this small estimation error, KF+SE still outperforms the KF.

Next, we show the effect of increasing the number of measurements ($n_y$, or equivalently $p$) on the estimation performance and consequently, the UAV's path following performance. 
This can be achieved through sensor fusion. 
For example, autonomous UAVs often use IMUs in addition to GPS for navigation, the former provides additional measurements such as the UAV's velocities, pitch and roll angles. 
Figure \ref{fig:ex_uav_traj} shows that increasing the number of measurements has no effect on the KF's estimation accuracy and hence, its path following ability. 
Even when 8 measurements are used the UAV equipped with a KF still fails to follow the desired trajectory. 
On the other hand, increasing the number of measurements improves the estimation performance of the secure estimator SE and consequently the performance of the combined scheme KF+SE. Figure \ref{fig:ex_uav_traj} shows that when 5 and 8 measurements are used, the UAV can follow its original planned path perfectly (KF + SE $n_y=5$ and KF + SE $n_y=8$).

\section*{Conclusion}
In this paper, we consider the estimation problem for UAVs under adversarial cyber attack and propose a secure estimation based KF that is computationally efficient and makes no assumptions about the attack signal model. We demonstrate that our proposed secure estimator outperforms standard KF, using numerical examples of UAVs under adversarial cyber attacks. This is important not only for today's aviation system but also delivery systems with drones in the near future.

\bibliographystyle{IEEEtran}
\tiny
\bibliography{reference}

\begin{thebibliography}{10}
\providecommand{\url}[1]{#1}
\csname url@samestyle\endcsname
\providecommand{\newblock}{\relax}
\providecommand{\bibinfo}[2]{#2}
\providecommand{\BIBentrySTDinterwordspacing}{\spaceskip=0pt\relax}
\providecommand{\BIBentryALTinterwordstretchfactor}{4}
\providecommand{\BIBentryALTinterwordspacing}{\spaceskip=\fontdimen2\font plus
\BIBentryALTinterwordstretchfactor\fontdimen3\font minus
  \fontdimen4\font\relax}
\providecommand{\BIBforeignlanguage}[2]{{%
\expandafter\ifx\csname l@#1\endcsname\relax
\typeout{** WARNING: IEEEtran.bst: No hyphenation pattern has been}%
\typeout{** loaded for the language `#1'. Using the pattern for}%
\typeout{** the default language instead.}%
\else
\language=\csname l@#1\endcsname
\fi
#2}}
\providecommand{\BIBdecl}{\relax}
\BIBdecl

\bibitem{faa}
``Press release -- {DOT} and {FAA} propose new rules for small unmanned
  aircraft systems,''
  \url{http://www.faa.gov/news/press_releases/news_story.cfm/?newsId=18295},
  accessed: 2015-02-15.

\bibitem{Google}
``{G}oogle project wing,''
  \url{http://www.theatlantic.com/technology/archive/2014/08/inside-googles-secret-drone-delivery-program/379306/?single_page=true},
  accessed: 2014-08-28.

\bibitem{Amazon}
``{A}mazon {P}rime {A}ir,'' \url{http://www.amazon.com/b?node=8037720011}.

\bibitem{Asctec}
``{Ascending Technologies},''
  \url{http://www.asctec.de/en/drone-uav/uav-uas-drone-powerline-infrastructure-inspection/}.

\bibitem{Croplife}
``{UAS}: the future of precision agriculture,''
  \url{http://www.croplife.com/equipment/precision-ag/uas-the-future-of-precision-agriculture/}.

\bibitem{Bullo}
F.~Pasqualetti, F.~Dorfler, and F.~Bullo, ``Cyber-physical attacks in power
  networks: models, fundamental limitations and monitor design,'' \emph{50th
  IEEE conference on decision and control and european control conference}, pp.
  2195 -- 2201, December 2011.

\bibitem{Liu}
K.~Manandhar, X.~Cao, F.~Hu, and Y.~Liu, ``Combating false data injection
  attacks in smart grid using kalman filter,'' \emph{International Conference
  on Computing, Networking and Communications}, pp. 16--20, February 2014.

\bibitem{Wu}
S.~Roy, C.~Ellis, S.~Shiva, D.~Dasgupta, V.~Shandilya, and Q.~Wu, ``A survey of
  game theory as applied to network security,'' \emph{43rd Hawaii International
  Conference on System Sciences}, 2010.

\bibitem{Basar}
A.~Gupta, C.~Langbort, and T.~Basar, ``Optimal control in the presence of an
  intelligent jammer with limited actions,'' \emph{49th IEEE Conference on
  Decision and Control}, pp. 1096 -- 1101, December 2010.

\bibitem{Basar2}
M.~H. Manshaei, Q.~Zhu, T.~Alpcan, T.~Basar, and J.-P. Hubaux, ``Game theory
  meets network security and privacy,'' \emph{ACM Computing Surveys}, vol.~45,
  no.~3, June 2013.

\bibitem{Walrand}
A.~Gueye, V.~Marbukh, and J.~C. Walrand, ``Towards a metric for communication
  network vulnerability to attacks: A game theoretic approach,'' \emph{3rd
  International ICST Conference on Game Theory for Networks}, May 2012.

\bibitem{Pappas}
M.~Fei, M.~Pajic, and G.~J. Pappas, ``Stochastic game approach for replay
  attack detection,'' \emph{52nd IEEE Conference on Decision and Control}, pp.
  1854 -- 1859, December 2013.

\bibitem{Fawzi2014}
H.~Fawzi, P.~Tabuada, and S.~Diggavi, ``Secure estimation and control for
  cyber-physical systems under adversarial attacks,'' \emph{Automatic Control,
  IEEE Transactions on}, vol.~59, no.~6, pp. 1454--1467, June 2014.

\bibitem{Candes_Tao}
E.~Candes and T.~Tao, ``Decoding by linear programming,'' \emph{Information
  Theory, IEEE Transactions on}, vol.~51, no.~12, pp. 4203--4215, Dec 2005.

\bibitem{David_Chang}
D.~Hayden, Y.~H. Chang, J.~Goncalves, and C.~Tomlin, ``Sparse network
  identifiability via compressed sensing,'' \emph{Automatica}, vol.~52, 2016.

\bibitem{Chang_Qie}
Y.~H. Chang, Q.~Hu, and C.~Tomlin, ``Secure estimation based {Kalman Filter for
  Cyber-Physical Systems} against adversarial attacks,''
  \emph{arXiv:1512.03853v2}, 2015.

\bibitem{KwonACC}
C.~Kwon, W.~Liu, and I.~Hwang, ``Security analysis for cyber-physical systems
  against stealthy deception attacks,'' \emph{American Control Conference},
  2013.

\bibitem{Bouffard}
P.~Bouffard, ``On-board model predictive control of a quadrotor helicopter:
  Design, implementation, and experiments,'' University of California Berkeley,
  http://www.eecs.berkeley.edu/Pubs/TechRpts/2012/EECS-2012-241.html, Technical
  Report UCB/EECS-2012-241, December 2012.

\end{thebibliography}
\bigskip

%\section*{Contact Author Email Address}
%\noindent \small {Mailto: qiehu@berkeley.edu}

\section*{Copyright Statement}

\noindent \small {The authors confirm that they, and/or their company or organization, hold copyright on all of the original material included in this paper. The authors also confirm that they have obtained permission, from the copyright holder of any third party material included in this paper, to publish it as part of their paper. The authors confirm that they give permission, or have obtained permission from the copyright holder of this paper, for the publication and distribution of this paper as part of the ICAS 2016 proceedings or as individual off-prints from the proceedings.}

\end{document}